\documentclass[11pt]{llncs}

\usepackage{fullpage}
\usepackage{url}
\usepackage{amsmath,amssymb}
\usepackage{graphicx}
\usepackage{epstopdf}
\usepackage {pstricks,pst-node}

\def\max{\operatorname{max}}

\begin{document}

\mainmatter 

\title{Minimal Conflicting Sets for the Consecutive Ones Property in
  ancestral genome reconstruction}

\author{Cedric Chauve\inst{1}
  \and Utz-Uwe Haus\inst{2} \and Tamon Stephen\inst{1} 
  \and Vivija P. You\inst{1}\thanks{Current address: Department of
  Mathematics, WestVirginia University, USA,
  \email{vpy@math.wvu.edu}}}

\institute{Department of Mathematics, Simon Fraser University,
  Burnaby (BC), Canada \\
  \email{[cedric.chauve,tamon,vpy]@sfu.ca} \and Institute for
  Mathematical Optimization, University of Magdeburg, Germany \\
  \email{haus@imo.math.uni-magdeburg.de}}


\maketitle

\begin{abstract}
  A binary matrix has the Consecutive Ones Property (C1P) if its
  columns can be ordered in such a way that all 1's on each row are
  consecutive.  A {\em Minimal Conflicting Set} is a set of rows that
  does not have the C1P, but every proper subset has the C1P. Such
  submatrices have been considered in comparative genomics
  applications, but very little is known about their combinatorial
  structure and efficient algorithms to compute them. We first
  describe an algorithm that detects rows that belong to Minimal
  Conflicting Sets. This algorithm has a polynomial time complexity
  when the number of $1$s in each row of the considered matrix is
  bounded by a constant. Next, we show that the problem of computing
  all Minimal Conflicting Sets can be reduced to the joint generation
  of all minimal true clauses and maximal false clauses for some
  monotone boolean function. We use these methods on simulated data
  related to ancestral genome reconstruction to show that computing
  Minimal Conflicting Set is useful in discriminating between true
  positive and false positive ancestral syntenies. We also study a
  dataset of yeast genomes and address the reliability of an ancestral
  genome proposal of the {\em Saccahromycetaceae} yeasts.

  \smallskip\noindent Draft, do not distribute. Version of \today.
  \end{abstract}


\section{Introduction}
\label{sec:intro}

A binary matrix $M$ has the Consecutive Ones Property (C1P) if its
columns can be ordered in such a way that all 1's on each row are
consecutive. Algorithmic questions related to the C1P for binary
matrices are central in genomics, for problems such as physical
mapping~\cite{alizadeh-physical,christof-branch,lu-test} and ancestral
genome reconstruction
(see~\cite{adam-modelfree,chauve-methodological,ma-reconstructing,ouangraoua-prediction}
for recent references).  Here we are interested in the problem of
inferring the architecture of an ancestral genome from the comparison 
of extant genomes (due to DNA decay, for example, such genomes cannot
be sequenced directly). 
Note however that our results are of interest for physical
mapping too.  Briefly, when inferring an ancestral genome architecture
from the comparison of extant genomes, it is common to represent
partial information about the ancestral genome $G$ as a binary matrix
$M$: columns represent genomic markers that are believed to have been
present in $G$, rows of $M$ represent groups of markers that are
believed to be co-localized in $G$, and the goal is to infer the order
of the markers on the chromosomes of $G$. Such ordering of the markers
define chromosomal segments called Contiguous Ancestral Regions
(CARs).

If the matrix $M$ contains only correct information (i.e. groups of
markers that were co-localized in the ancestral genome of interest),
then it has the C1P, which can be decided in linear-time and
space~\cite{booth-testing,habib-lex,hsu-simple,mcconnell-certifying,meidanis-on}.
However, with most real datasets, $M$ contains errors. These can be
either incorrect columns, that represent genomic markers that were not
present in $G$, or incorrect rows, that represent groups of markers that
were not co-localized in $G$\footnote{Note however that this
  classification of possible errors is somewhat simplified
  (see~\cite{goldberg-four} for a more detailed discussion regarding
  errors in physical mapping).}. A fundamental question is then to
detect such errors in order to correct $M$, and the classical approach
to handle these (unknown) errors rely on combinatorial optimization,
asking for an optimal transformation of $M$ into a matrix that has the
C1P, for some notion of transformation of a matrix linked to the
expected errors; for example, if incorrect markers (resp. groups of
co-localized genes) are expected, one could ask for a maximal subset
of columns (resp. rows) of $M$ that has the C1P. In both cases, such
combinatorial optimization problems are intractable
(see~\cite{dom-recognition,dom-algorithmic} for recent surveys).

In the present work, we assume the following situation: $M$ is a
binary matrix that represents information about an unknown ancestral
genome $G$ and does not have the C1P due to erroneous rows, from now
called {\em false positives}. The notion of {\em Minimal Conflicting
  Set} was introduced to handle non-C1P binary matrices and false
positives in~\cite{bergeron-reconstructing}
and~\cite{stoye-unified}. If a binary matrix $M$ does not have the
C1P, a {\em Minimal Conflicting Set} (MCS) is a submatrix $M'$ of $M$
composed of a subset of the rows of $M$ such that $M'$ does not have
the C1P, but every proper subset of rows of $M'$ has the C1P. The
Conflicting Index (CI) of a row of $M$ is the number of MCS it belongs
to. Hence, MCS can be seen as the smallest structures that prevent a
matrix from having the C1P. It is then natural to expect that false
positive belong to MCS, and that every MCS contains at least one false
positive. In~\cite{bergeron-reconstructing}, an extreme approach was
followed in handling non-C1P matrices: all rows belonging to at least
one MCS were discarded from $M$, which can consequently
discard also a large number of true
positives. In~\cite{stoye-unified}, rows were ranked according to
their CI (or more precisely an approximation of their CI) before being
processed by a branch-and-bound algorithm to extract a maximal subset
of rows of $M$ that has the C1P. These two approaches raise natural
algorithmic questions related to MCS that we address here:
\begin{itemize}
\item For a row $r$ of $M$, is the CI of $r$ greater than $0$?
\item Can we compute the CI of all rows $r$ of $M$ or enumerate all MCS of
  $M$?
\end{itemize}
Our work is motivated by the fact that the fundamental question is to
detect the false positives rows in $M$ rather than extracting a
maximal C1P submatrix. We investigate here, using both simulations and
real data, the following question: does a false positive row have some
characteristic properties in terms of MCS or CI? This question
naturally extends to the notion of {\em Maximal C1P Sets} (MC1PS), the
dual notion of MCS, that represent sets of row that do have the C1P
but can not be extended while maintaining this property. 

After some preliminaries on the C1P, MCS and MC1PS
(Section~\ref{sec:prelim}), we attack two problems. First, in
Section~\ref{sec:VPY}, we consider the problem of deciding if a given
row of a binary matrix $M$ belongs to at least one MCS.  We show that,
when all rows of a matrix are constrained to have a bounded number of
$1$'s, deciding if the CI of a row of a matrix is greater than $0$ can
be done in polynomial time.  The constraint on the number of $1$'s per
row is motivated by real applications: in~\cite{ma-reconstructing} for
example, adjacencies, that is rows with two $1$s per row, were
considered for the reconstruction of an ancestral mammalian genome.
Next, in Section~\ref{sec:MBF}, we attack the problem of generating
all MCS or MC1PS for a binary matrix $M$. We show that this problem
can be approached as a joint generation problem of minimal true
clauses and maximal false clauses for monotone boolean functions.
This can be done in quasi-polynomial time thanks to an oracle-based
algorithm for the dualization of monotone boolean
functions~\cite{eiter-computational,fredman-complexity,gurvich-generating}.
We implemented this algorithm~\cite{cl-jointgen} and applied it on
simulated and real data (Section~\ref{sec:results}).  Application on
simulated data suggest that the computing all conflicting sets and the
conflicting index of all rows of a binary matrix is useful to
discriminate between true positive and false positive ancestral
syntenies.  We also study a real dataset of yeast genomes and address
the reliability of an ancestral genome proposal of the {\em
  Saccahromycetaceae} yeasts. We conclude by discussing several open
problems.

\section{Preliminaries}
\label{sec:prelim}

We briefly review here ancestral genome reconstruction and known
algorithmic results related to Minimal Conflicting Sets and Maximal
C1P Sets.

\subsection{The Consecutive Ones Property, Minimal Conflicting Sets, Maximal C1P Sets}\label{sec:C1P}
Let $M$ be a binary matrix with $m$ rows and $n$ columns, with $e$
entries $1$. We denote by $r_1,\ldots,r_m$ the rows of $M$ and
$c_1,\ldots,c_n$ its columns. We assume that $M$ does not have two
identical rows, nor two identical columns, nor a row with less than
two entries $1$ or a column with no entry $1$. We denote by
$\Delta(M)$ the maximum number of entries $1$ found in a single row of
$M$, called the {\em degree} of $M$.  In the following, we sometimes
identify a row of $M$ with the set of columns where it has entries
$1$, and a set of rows with the matrix defined by the submatrix of $M$
containing exactly these rows.

\begin{definition}\em\label{def:MCS}
A Minimal Conflicting Set (MCS) is a set $R$ of rows of $M$ that does
not have the C1P but such that every proper subset of $R$ has the C1P.
The {\em Conflicting Index} (CI) of a row $r_i$ of $M$ is the number
of MCS that contain $r_i$. A row $r_i$ of $M$ that belongs to at least
one conflicting set is said to be a {\em conflicting row}. The {\em
  Conflicting Ratio} (CR) of a row $r_i$ of $M$ is the ratio between
the CI of $r_i$ and the number of MCS that $M$ contains.

\end{definition}

\begin{definition}\em\label{def:MC1PS}
A Maximal C1P Set (MC1PS) is a set $R$ of rows of $M$ that has the C1P
and such that adding any additional row from $M$ to it results in a
set of rows that does not have the C1P. The {\em MC1PS Index} (C1PI)
of a row $r_i$ of $M$ is the number of MC1PS that contain $r_i$. The
{\em MC1PS Ratio} (C1PR) of a row $r_i$ of $M$ is the ratio between the
C1PI of $r_i$ and the number of MC1PS that $M$ contains.
\end{definition}

For a subset $I=\{i_1,\ldots,i_k\}$ of $[n]$, we denote by $R_I$ the
set $\{r_{i_1},\ldots,r_{i_k}\}$ of rows of $M$. If $R_I$ is an MCS
(resp.~MC1PS), we then say that $I$ is an MCS (resp.~MC1PS).

\subsection{Ancestral genome reconstruction}\label{sec:AGR} The present work is
motivated by the problem of inferring an ancestral genome architecture
given a set of extant genomes. An approach to this problem, described
in~\cite{chauve-methodological}, consists in defining an alphabet of
genomic markers that are believed to appear uniquely in the extinct
ancestral genome. An {\em ancestral synteny} is a set of markers that
are believed to have been consecutive along a chromosome of the
ancestor. A set of ancestral syntenies can then be represented by a
binary matrix $M$\/: columns represent markers and the $1$ entries of
a given row define an ancestral synteny. If all ancestral syntenies
are true positives (i.e. represent sets of markers that were
consecutive in the ancestor), then $M$ has the C1P and defines a set
of {\em Contiguous Ancestral Regions} (CARs),
see~\cite{chauve-methodological,ma-reconstructing}. Otherwise, some
ancestral syntenies are false positives that create MCS and the key
problem is to detect and discard them. 

Note however that we do not assume that any false positive creates an
MCS; indeed it is possible that a false positive row
contains two entries $1$ corresponding to two markers that are
extremities of two real ancestral chromosomes, and then this false
positive would create a chimeric ancestral chromosome. 
Such false positives have to be detected using techniques
other than the ones we describe in the present work, as
there is no combinatorial signal to detect them based on
the Consecutive Ones Property.

In the framework described in~\cite{chauve-methodological} (that was
also used in~\cite{adam-modelfree,ma-reconstructing,stoye-unified}),
ancestral syntenies are defined as common intervals (markers consist 
of two genome segments having the same content,
see~\cite{bergeron-formal}) of a pair of genomes whose evolutionary
path goes through the desired ancestor. As ancestral syntenies are
detected by mining common intervals in pairs of genomes, it is then
easy to control the degree of the resulting matrix by restricting the
comparison to common intervals of bounded size, such as adjacencies
(ancestral syntenies of size 2), which is fundamental for the result
we describe in Section~\ref{sec:VPY}.

Also, it is common to weight the rows of $M$ with a measure of
confidence in its quality, for example based on the distribution of a
group of co-localized markers among the considered extant
genomes~\cite{chauve-methodological,ma-reconstructing}, and we can
expect that false positives have a lower score in general than true
positives. However, the concepts of MCS and MC1PS are independent of this
weighting and we do not consider it here from a theoretical point of
view, although we do consider it in our experiments on real data.

\subsection{Preliminary algorithmic results on MCS and MC1PS}
\label{sec:prelim-algo}

In the case where each row of $M$ has exactly two entries $1$, $M$
naturally defines a graph $G_M$ with vertex set $\{c_1,\ldots,c_n\}$
and where there is an edge between $c_i$ and $c_j$ if and only if
there is a row with entries $1$ in columns $c_i$ and $c_j$. The
following property is then obvious.

\begin{property}\label{prpy:MCS-2}
  If $\Delta(M)=2$, a set of rows $R$ of $M$ is an MCS 
  if and only if the subgraph induced by the corresponding edges is a
  star with four vertices (also called a {\em claw}) or a cycle.
\end{property}

This property implies immediately that both the number of MCS and the
number of MC1P can be exponential in $n$. Also, combined with the fact
that counting the number of cycles that contain a given edge in an
arbitrary graph is \#P-hard~\cite{valiant-complexity}, this leads to
the following result.

\begin{theorem}
  \label{thm:CI-complexity}
  The problem of computing the Conflicting Index of a row in a binary
  matrix is \#P-hard.
\end{theorem}

Given a set of $p$ rows $R$ of $M$, deciding whether these rows form
an MCS can be achieved in polynomial time by testing (1) whether they
form a matrix that does not have the C1P and, (2) whether every
maximal proper subset of $R$ (obtained by removing exactly one row)
forms a matrix that has the C1P. This requires only $p+1$ C1P tests
and can then be done in time $O(p(n+p+e))$, using an efficient
algorithm for testing the C1P~\cite{mcconnell-certifying}.

The problem of generating one MCS is not hard and can be achieved in
polynomial time by the following simple greedy algorithm:
\begin{enumerate}
\item let $R=\{r_1,\ldots,r_m\}$ be the complete set of rows of $M$;
\item
  for $i$ from $1$ to $m$, if removing $r_i$ from $R$ results in a set
  of rows that has the C1P then keep $r_i$ in $R$, otherwise remove
  $r_i$ from $R$; 
\item the subset of rows $R$ obtained at the end of this loop is then
  an MCS.
\end{enumerate}

Given $M$ and a list $C=\{R_1,$ $\ldots,R_k\}$ of known MCS, the {\em
  sequential generation problem} $Gen_{MCS}(M, C)$ is the following:
decide if $C$ contains all MCS of $M$ and, if not, compute one MCS
that does not belong to $C$. Using the obvious property that, if $R_i$
and $R_j$ are two MCS then neither $R_i \subset R_j$ nor $R_j \subset
R_i$, Stoye and Wittler~\cite{stoye-unified} proposed the following
backtracking algorithm for $Gen_{MCS}(M,C)$:
\begin{enumerate}
\item Let $M'$ be defined by removing from $M$ at least one row from
  each $R_i$ in $C$ by recursing on the elements of $C$.
\item If $M'$ does not have the C1P then compute an MCS of $M'$ and
  add it to $C$, else backtrack to step 1 using another set of rows to
  remove such that each $R_i\in C$ contains at least one of these
  rows.
\end{enumerate}
This algorithm can require time $\Omega(n^k)$ to terminate, which, as
$k$ can be exponential in $n$, can be superexponential in $n$.  As far
as we know, this is the only previously proposed algorithm to compute
all MCS. 

\begin{remark}\label{prem-MC1P}
  The algorithms to decide if a set of rows is an MCS, compute an MCS
  or compute all MCS can be transformed in a straightforward way to
  answer the same questions for MC1PS, with similar complexity. The
  analogue of Property~\ref{prpy:MCS-2} for MC1PS in matrices of
  degree $2$ is that a MC1P is a maximal set of paths: adding an edge
  creates a claw or a cycle. We are not aware of any result on
  counting or enumerating maximal sets of paths.
\end{remark}

In summary, 
the number of MCS or MC1PS can be exponential, and there
is no known efficient algorithm to decide, in general, if a given row
belongs to some MCS.  In the next section we show that there is an
efficient algorithm if $\Delta(M)$ is fixed.

\section{Deciding if a row is a conflicting row.}
\label{sec:VPY}

We now describe our first result, an algorithm to decide if a row of
$M$ is a conflicting row (i.e. has a CI greater than $0$). Detecting
non-conflicting rows is important, for example to speed-up algorithms
that compute an optimal C1P subset of rows of $M$, or in generating
all MCS. Our algorithm has a complexity that is exponential in
$\Delta(M)$. It is based on a combinatorial characterization of
non-C1P matrices due to Tucker~\cite{tucker-structure}.

\paragraph{Tucker patterns.} The class of C1P matrices is closed under
column and row deletion.  Hence there exists a characterization of
matrices which do not have the C1P by forbidden minors.
Tucker~\cite{tucker-structure} characterizes these forbidden
submatrices, called $M_{I}$, $M_{{II}}$ $M_{{III}}$, $M_{IV}$ and
$M_V$: if $M$ is binary matrix that does not have the C1P, then it
contains at least one of these matrices as a submatrix. We call these
forbidden matrices the {\em Tucker patterns}.  Patterns $M_{IV}$ and
$M_{V}$ each have $4$ rows and respectively $6$ and $5$ columns, while
$M_{I}, M_{II}, M_{III}$ are $(q+2)$ by $(q+2)$, $(q+3)$ by $(q+3)$
and $(q+2)$ by $(q+3)$ respectively for a parameter $q \ge 1$.
When $q=1$, pattern $M_I$ corresponds to cycle and pattern $M_{III}$
corresponds to the claw.  The patterns are described in the Appendix.

\paragraph{Bounded patterns.}
Let $P$ be a set of $p$ rows $R$ of $M$, that defines a $p\times n$
binary matrix. $P$ is said to {\it contain exactly} a Tucker pattern
$M_X$ if a subset of its columns defines a matrix equal to pattern
$M_X$.  The following properties are straightforward from the
definitions of Tucker patterns and MCS:

\begin{property}
  \label{prpy:tucker}
  (1) A set $P$ of rows of $M$ is an MCS if and only if $P$ contains
  exactly a Tucker pattern and no proper subset of rows of $P$ does
  contain exactly a Tucker pattern.  
  \\   
  (2) If a subset of $p$ rows of $M$ contains exactly a Tucker pattern
  $M_{{II}}$ $M_{{III}}$, $M_{IV}$ or $M_{V}$, then $4 \leq p \leq
  \max(4,\Delta(M)+1)$.
\end{property}

If a row $r_i$ satisfies the conditions of
Property~\ref{prpy:tucker}.(1) for a Tucker pattern $M_X$, we say that
$r_i$ belongs to an MCS {\em due to pattern $M_X$}. Tucker patterns
with at most $\Delta(M)+1$ rows are said to be {\em
  bounded}. Property~\ref{prpy:tucker} leads to the following results.

\begin{proposition}
  \label{prop:bounded}
  Let $M$ be a binary matrix that does not have the C1P, and
  $r_i$ a row of $M$. Deciding if $r_i$ belongs to an MCS due to a Tucker pattern
  of $p$ rows can be done in $O(m^{p-1}p(n+p+e))$ worst-case
  time.
\end{proposition}

\begin{proof} 
  The following brute-force algorithm decides if $r_i$
  belongs to an MCS due to a Tucker pattern of $p$ rows.
  \begin{itemize}
  \item
    Examine all sets of $p$ rows of $M$ that contain $r_i$.
  \item 
    For a given set of rows, if it does not have the C1P but every
    proper subset does have the C1P, then $r_i$ belongs to an
    MCS due to a Tucker pattern of $p$ rows. 
  \item 
    If no such set satisfies this property, then $r_i$ does not
    belong to any MCS due to a Tucker pattern of $p$ rows.
  \end{itemize}
  The complexity can be explained as follows: there are
  $O(m^{\max(3,p-1)})$ subsets of the $m$ rows of $M$ to
  consider. For each such set, we need to perform at most $p+1$ C1P
  tests, and each such test can be performed in $O(n+p+e)$ worst-case
  time.
  \qed
\end{proof}

The following corollary follows immediately from
Property~\ref{prpy:tucker}.(2).

\begin{corollary}
  \label{cor:bounded}
  Let $M$ be a binary matrix that does not have the C1P, and $r_i$ a
  row of $M$. Deciding if $r_i$ belongs to an MCS due to a bounded
  Tucker pattern can be done in
  $O(m^{\max(3,\Delta(M))}\Delta(M)(n+\Delta(M)+e))$ worst-case time.
\end{corollary}

\paragraph{Unbounded patterns.}
We now describe how to decide if a row $r_i$ of $M$ belongs to an MCS
due to an unbounded Tucker pattern. From
Property~\ref{prpy:tucker}.(2), this Tucker pattern can only be a
pattern $M_I$ with at least $\Delta(M)+2$ rows. The key idea is that
Tucker pattern $M_I$ describes a cycle in a bipartite graph encoded by
$M$.

Let $B_M$ be the bipartite graph defined by $M$ as follows: vertices
are rows and columns of $M$, and every entry $1$ in $M$ defines an
edge. Pattern $M_{I}$ with $p$ rows corresponds to a cycle of length
$2p$ in $B_M$. Hence, if $R$ contains $M_I$ with $p$-row, the subgraph
of $B_M$ induced by $R$ contains such a cycle and possibly other
edges.

Let $C=(r_{i_1}, c_{j_1} \ldots, r_{i_p}, c_{j_p})$ be a cycle
in $B_M$. We say that a vertex $r_{i_q}$ belonging to $C$ is {\em
  blocked in $C$} if there exists a vertex $c_j$ such $M_{i_q,j}=1$,
$M_{i_{q-1},j}=1$ (resp. $M_{i_p,j}=1$ if $q=1$) and
$M_{i_{q+1},j}=1$ (resp. $M_{1,j}=1$ if $q=p$). In other words,
replacing the path between $r_{i_{q-1}}$ and $r_{i_{q+1}}$ going
through $r_{i_q}$ in $C$ by the edges $\{r_{i_{q-1}},c_j\}$ and
$\{r_{i_{q+1}},c_j\}$ gives a shorter cycle that does not contain
$r_{i_q}$.

\begin{proposition}\label{prop:cycle}
  Let $M$ be a binary matrix that does not have the C1P and $r_i$ be a
  row of $M$ that does not belong to any MCS due to a bounded Tucker
  pattern. Then $r_i$ belongs to an MCS if and only if $r_i$ belongs
  to a cycle $C=(r_{i_1}, c_{j_1}, \ldots r_{i_p}, c_{j,p})$ in $B_M$,
  with $p \geq \Delta(M)+2$, and $r_i$ is not blocked in $C$.
\end{proposition}

\begin{proof}
  If $r_i$ belongs to an MCS but not due to a bounded Tucker pattern,
  then, according to Property~\ref{prpy:tucker}.(2), it is due to
  pattern $M_I$ defined on $p$ rows of $M$ containing $r_i$. So $r_i$
  belongs to a cycle $C$ defined by this pattern $M_I$. However, if
  $r_i$ is blocked in $C$, then, removing the row $r_i$ and its
  adjacent edges still leaves a cycle in the subgraph of $B_M$ induced
  by the remaining vertices, which contradicts the fact that the
  initial $p$ rows form an MCS.

  Now, assume that $r_i$ belongs to a cycle $C=(r_{i_1}, c_{j_1},
  \ldots r_{i_p}, c_{j,p})$ in $B_M$ and $r_i$ is not blocked in
  $C$. We want to show that $r_i$ belongs to an MCS. Assume moreover
  that $C$ is minimal in the following sense: there is no smaller
  cycle in $B_M$ containing $r_i$ unblocked.
  \begin{itemize}
  \item
    The set $P=\{r_{i_1},\ldots,r_{i_p}\}$ of rows obviously does not
    have the C1P, because it contains (exactly) the Tucker pattern
    $M_I$.
  \item
    We now want to show that removing any row from this gives a
    matrix that has the C1P. Let $r_{i_j}$ be a row belonging to
    $C$. Assume that removing $r_{i_j}$ and the adjacent edges results
    in a matrix with $p-1$ rows that does not have the C1P, and then
    contains an MCS.  If $r_i$ belongs to this MCS, then it is due to
    a Tucker pattern $M_I$, as by hypothesis it does not belong to an
    MCS due to a bounded Tucker pattern. This then contradicts the
    minimality assumption on the cycle $C$. Now assume that $r_i$ does
    not belong to this MCS, which is then included in a set $Q$ of
    $q=p-2$ or $q=p-1$ (if $r_i=r_{i_j}$) rows.  A {\em chord} in the
    cycle $C$ is a set of two edges $(r,c)$ and $(r',c)$ such that $r$
    and $r'$ belong to $C$ but are not consecutive row vertices in
    this cycle. We claim that the definition of Tucker patterns
    implies that there always exist a column in $M$ that did not
    belong to $C$ and that defines a chord in $C$ between two rows of
    $Q$. This again contradicts the minimality assumption on the cycle
    $C$. Hence, by contradiction, we have that removing $r_{i_j}$ and
    the adjacent edges results in a matrix that has the C1P.
  \end{itemize}
  \qed
\end{proof}

\paragraph{The algorithm.}  To decide whether $r_i$ belongs to an MCS,
we can then (1) decide whether it belongs to an MCS due to a bounded
pattern, as described in Corollary~\ref{cor:bounded}, and then, if
this is not the case, (2) check whether $r_i$ belongs to an MCS due to
an unbounded pattern. For this second case, we only need to find a
cycle where $r_i$ is not blocked.  This can be done in polynomial time
by considering all pairs of possible rows $r_{i_1}$ and $r_{i_2}$ that
each have an entry 1 in a column where $r_i$ has an entry 1 (there are
at most $O(m^2)$ such pairs of rows), exclude the cases where the
three rows $r_i$, $r_{i_1}$ and $r_{i_2}$ have an entry 1 in the same
column, and then check if there is a path in $B_M$ between $r_{i_1}$
and $r_{i_2}$ that does not visit $r_i$. This leads to the main result
of this section.

\begin{theorem}\label{thm:main}
  Let $M$ be an $m\times n$ binary matrix that does not have the C1P,
  and $r_i$ be a row of $M$.  Deciding if $r_i$ belongs to at least
  one MCS can be done in
  $O(m^{\max(3,\Delta(M))}\Delta(M)(n+\Delta(M)+e))$ time.
\end{theorem}

\section{Generating all MCS and MC1PS using Monotone Boolean Functions}
\label{sec:MBF}

In this section, we describe an algorithm that enumerates all MCS and
MC1PS of a binary matrix $M$ simultaneously in quasi-polynomial
time. They key point is to describe this generation problem as a joint
generation problem for monotone boolean functions.

Let $[m]=\{1,2,\ldots,m\}$.  For a set $I=\{i_1,\ldots,i_k\} \subseteq
[m]$, we denote by $X_I$ the boolean vector $(x_1,\ldots,x_m)$ such
that $x_j=1$ if and only if $i_j \in I$.

\begin{definition}\em\label{def:MBF}
 A {\em boolean function} $f: \{0,1\}^m \rightarrow \{0,1\}$ is said
 to be {\em monotone} if for every $I,J \subseteq [m]$, $I \subseteq J
 \Rightarrow f(X_I) \leq f(X_J)$.
\end{definition}

\begin{definition}\em\label{def:MTC-MFC}
  Given a boolean function $f$, a boolean vector $X$ is said to be a
  {\em Minimal True Clause} (MTC) if $f(X_I)=1$ and $f(X_J)=0$ for
  every $J \subset I$.  Symmetrically, $X_I$ is said to be a {\em
    Maximal False Clause} (MFC) if $f(X_I)=0$ and $f(X_J)=1$ for every
  $I \subset J$.  We denote by $MTC(f)$ (resp.~$MFC(f)$) the set of
  all MTC (resp.~MFC) of $f$.
\end{definition}

For a given $m\times n$ binary matrix $M$, let $f_M: \{0,1\}^m
\rightarrow \{0,1\}$ be the boolean function defined by $f_M(X_I)=1$
if and only if $R_I$ does not have the C1P, where $I \subseteq [m]$.
This boolean function is obviously monotone and the following
proposition is immediate.

\begin{proposition}
  \label{prop:C1P-MBF}
  Let $I=\{i_1,\ldots,i_k\} \subseteq [m]$. $R_I$ is an MCS (resp.~MC1PS) 
  of $M$ if and only if $X_I$ is an MTC (resp.~MFC) for $f_M$.
\end{proposition}

It follows from Proposition~\ref{prop:C1P-MBF} that generating all MCS
reduces to generating all MTC for a monotone boolean function. This
very general problem has been the subject of intense research, and we
describe briefly below some important properties. 

\begin{theorem}{\em \cite{gurvich-generating}}
  \label{thm:MTC-seq}
  Let $C=\{X_1, \ldots, X_k\}$ be a set of MTC (resp.~MFC) of a
  monotone boolean function $f$.  The problem of deciding if $C$
  contains all MTC (resp.~MFC) of $f$ is coNP-complete.
\end{theorem}

\begin{theorem}{\em \cite{gunopulos-data}}
  \label{thm:MTC-oracle}
  The problem of generating all MTC of a monotone boolean function $f$
  using an oracle to evaluate this function can require up to
  $|MTC(f)+MFC(f)|$ calls to this oracle.
\end{theorem}

This property suggests that, in general, to generate all MTC, it
is necessary to generate all MFC, and vice-versa.  For example, the
algorithm of Stoye and Wittler \cite{stoye-unified} described in
Section~\ref{sec:prelim} is a satisfiability oracle based algorithm --
it uses a polyno\-mial-time oracle to decide if a given submatrix has
the C1P, but it doesn't use this structure any further. Once it has
found the complete list $C$ of MCS, it will proceed to check all MC1PS
sets as candidate conflicting sets before terminating. Since this does
not keep the MC1PS sets explicitly, but instead uses backtracking, it
may generate the same candidates repeatedly resulting in a substantial
duplication of effort. In fact, this algorithm can easily be modified
to produce {\it any} monotone boolean function given by a truth
oracle.

One of the major results on generating MTC for monotone boolean
functions, is due to Fredman and Khachiyan. It states that generating
both sets together can be achieved in time quasi-polynomial in the
number of MTC plus the number of MFC.

\begin{theorem}{\em \cite{fredman-complexity}}
  \label{thm:MTC-MFC}
  Let $f: \{0,1\}^m \rightarrow \{0,1\}$ be a monotone boolean
  function whose value at any point $x \in \{0,1\}^m$ can be
  determined in time $t$, and let $C$ and $D$ be respectively the sets
  of the MTC and MFC of $f$. Given two subsets $C'\subseteq C$ and
  $D'\subseteq D$ of total size $s=|C'|+|D'|$, deciding if $C\cup
  D=C'\cup D'$, and if $C\cup D\neq C'\cup D'$ finding an element in
  $(C\backslash C') \cup (D\backslash D')$ can be done in time
  $O(m(m+t)+s^{o(\log s)})$.
\end{theorem}
  
The key element to achieve this result is an algorithm that tests if
two monotone boolean functions are duals of each other
(see~\cite{eiter-computational} for a recent survey on this topic). As
a consequence, we can then use the algorithm of Fredman and Khachiyan
to generate all MCS and MC1PS in quasipolynomial time.

\section{Experimental results}
\label{sec:results}

We present here results obtained on simulated and real datasets 
using
the \texttt{cl-jointgen} implementation (release 2008-12-01)
of the joint generation method
which is publicly available~\cite{cl-jointgen} with an oracle to test
the C1P property based on the algorithm described
in~\cite{mcconnell-certifying}.

\subsection{Simulated data}
\label{sec:simulations}

We generated several simulated datasets of ancestral syntenies as follows:
\begin{itemize}
\item We started from an ancestral unichromosomal genome $G$ composed
  of $40$ genomic markers, labeled from $1$ to $40$ and ordered
  increasingly along this chromosome (i.e. it is represented by the
  identity permutation on $\{1, \ldots, 40\}$.
\item From this ancestor, we extracted $39$ true positive ancestral
  syntenies, labeled $a_1,\ldots,a_{39}$, in such a way that $a_i$ is
  an interval of $G$ starting at marker $i$ and of length chosen
  uniformly in the set $\{2,\ldots,d\}$, where $d$ is a parameter of
  the dataset corresponding to the maximum degree of the generated
  matrix. We considered the value $d=2,3,4,5$. 
\item Finally, we added $6$ false positive ancestral syntenies,
  defined as sets of markers containing between $2 \leq c \leq d$
  markers and spanning an interval of $G$ of at most $g$ markers
  (hence the gaps in this false positive contain $g-c$ markers), where
  $g$ is a parameter of the dataset. We considered the values $g=2,5,10$.
\item To simulate the fact that, in general, if ancestral syntenies
  are weighted according to their conservation in extant species,
  false positives have a lesser weight than true positives, we
  weighted false positives by $0.5$ and true positives by $1.0$.
\item For each pair of parameters $(d,g)$, we generated $10$ datasets.
\end{itemize}

This generation method was designed to simulate moderately large
datasets that resemble real datasets. For most of the $120$ datasets,
the generation of all MCS and MC1PS could be completed within three
hours of computation, but for $13$ of them (one for $(d,g)=(3,5)$,
three for $(4,5)$, one for $(5,5)$, four for $(3,10)$, three for
$(4,10)$ and two for $(5,10)$) that were stopped if the computations
were not completed after three hours.  These unfinished datasets were
discarded when computing the statistics described below.

Table~\ref{tab:simul1} presents a summary of the number of MCS and
MC1PS observed in all the completed datasets; for computations that
had to be interrupted, similar results are observed from the partial
information contained in the log files. We can first observe the
number of MC1PS is much larger than the number of MCS. The second
important observation is a general trend towards increasing the number
of MCS when either $d$ or $g$ increases, which can be explained by the
more intricate combinatorial structure of sets of ancestral
syntenies. Indeed, for example, with $d=2$ and $g=2$, MCS are easy to
find and count, as cycles in the corresponding bipartite graph are
short and are relatively easy to find, even using brute-force
approach. With larger values of $g$, cycles length increase and
overlapping cycles appear more frequently, which results in more
MCS. Although it seems natural that increasing the value of $d$
results in more MCS understanding more precisely the impact of
increasing $d$ requires a better understanding of the combinatorial
structure of MCS with matrices of degree larger than $2$
(see~\cite{you-on} for the case $d=3$). Regarding MC1PS, we notice
that increasing $g$ results also in an increase of the number of
MC1PS, but we also notice that when $d$ attains the value $5$, there
seems to be a decrease of the number of MC1PS. A possible explanation
is that, with such degree, constraints on sets of rows that have the
C1P increase as a given row can now overlaps a large number of other
rows, which reduces the number of MC1PS containing such
rows. Generally, the impact of the degree on MC1PS deserves further
theoretical or experimental investigations.

\begin{table}\small
  \begin{center}
    \begin{tabular}{|c|c|c|c|c|c|c|}
      \hline
      $(d,g)$ & Min. Number & Max. Number  & Average Number & Min. Number &
      Max. Number & Average Number
      \\ 
      & of MCS & of MCS & of MCS  &  of MC1PS &  of MC1PS & of MC1PS 
      \\ \hline \hline
      (2,2)       &  17    & 25  & 21.4  & 464 & 3120 & 1393.2
      \\ \hline
      (3,2)       &  27    & 51  & 38.8  & 492 & 4800 & 2701.4
      \\ \hline
      (4,2)       &  42    & 84  & 59.1  & 756 & 11286& 2861.8
      \\ \hline
      (5,2)       &  68    & 137 & 93.3  & 160 & 10320& 2757.2 
      \\ \hline \hline
      (2,5)       &  16    & 29  & 22.4  & 1360& 10800& 4927.9
      \\ \hline
      (3,5)       &  31    & 106 & 58.1  & 624 & 10395& 5431.9
      \\ \hline
      (4,5)       &  65    & 176 & 107.7 & 1612& 11934& 6366.9
      \\ \hline 
      (5,5)       &  75    & 356 & 154.8 & 1785& 9420 & 3898.3
      \\ \hline \hline
      (2,10)      &  22    & 46  & 32.4  & 601 & 16954& 5796.8
      \\ \hline
      (3,10)      &  72    & 133 & 94.3  & 3876& 16030& 10995.8
      \\ \hline
      (4,10)      &  101   & 583 & 265.3 & 1434& 23474& 13278.3
      \\ \hline
      (5,10)      &  263   & 487 & 310.9 & 5432& 12362& 8909.5
      \\ \hline \hline
      all values  &  16    & 583 & 96.7  & 160 & 23474& 5312.5
      \\ \hline
    \end{tabular}
    \caption{Statistics on the number of MCS and MC1PS in all completed datasets.}
    \label{tab:simul1}
  \end{center}
\end{table}

For each dataset, and each ancestral synteny, we also computed two
statistics, the Conflicting Ratio (CR) and the MC1PS Ratio (C1PR). For
every row, we also computed its {\em MCS rank} and {\em MC1PS rank},
defined as follows: the MCS (resp. MC1PS) rank of a row of $M$ is its
rank when rows are ordered by increasing CR (resp. increasing C1PR).
Table~\ref{tab:simul2} presents a summary of these statistics. We can
notice that, in general, MCS seem to discriminate slightly better
between false positives and true positives in terms of ratio: the
average difference between the CR of a false positive and of a true
positive is slightly larger than the average difference of the
C1PR. The difference between the rankings is more strongly in favor of
MCS: on average, a false positive is more likely to have a MCS rank
close to the maximum rank than to have a low MC1PS rank. The
conclusion we can draw from these average results is that the CR and
MCS rank seem to better discriminate false positive from true
positives\footnote{The opposite conclusion was stated in the
  preliminary version of this paper~\cite{chauve-v0}, due to an
  experimental error.}. Considering the weights of the rows to weight
MCS and C1PS does not change significantly this conclusion (results
not shown).

\begin{table}\small
  \begin{center}
    \begin{tabular}{|c|c|c|c|c|c|c|}
      \hline
      Dataset & Average & Average & Average & Average & Average FP & Average FP
      \\ 
              & FP\_CR  &  TP\_CR &  FP\_C1PR & TP\_C1PR  & MCS rank   & MC1PS rank
      \\ \hline \hline
      (2,2)  & 0.20 & 0.05 & 0.66 & 0.82 & 39.65 & 15.47
      \\ \hline
      (3,2)  & 0.20 & 0.06 & 0.67 & 0.76 & 38.48 & 19.27
      \\ \hline
      (4,2)  & 0.21 & 0.06 & 0.61 & 0.73 & 39.35 & 18.45
      \\ \hline
      (5,2)  & 0.21 & 0.05 & 0.59 & 0.69 & 38.37 & 19.57
      \\ \hline \hline
      (2,5)  & 0.21 & 0.07 & 0.69 & 0.79 & 39.53 & 17.72
      \\ \hline
      (3,5)  & 0.21 & 0.07 & 0.63 & 0.73 & 38.33 & 18.72
      \\ \hline
      (4,5)  & 0.21 & 0.06 & 0.60 & 0.65 & 38.16 & 21.19
      \\ \hline 
      (5,5)  & 0.21 & 0.06 & 0.59 & 0.65 & 37.91 & 21.33
      \\ \hline \hline
      (2,10) & 0.27 & 0.11 & 0.62 & 0.81 & 38.20 & 14.02
      \\ \hline
      (3,10) & 0.23 & 0.08 & 0.62 & 0.69 & 36.41 & 19.92
      \\ \hline
      (4,10) & 0.26 & 0.07 & 0.55 & 0.62 & 39.52 & 20.19
      \\ \hline
      (5,10) & 0.23 & 0.07 & 0.55 & 0.58 & 37.15 & 22.38
      \\ \hline
    \end{tabular}
    \caption{Statistics on MCS and MC1PS on simulated datasets. FP\_CR
      is the Conflicting Ratio for False Positives, TP\_CR is for CR
      the True Positives, FP\_MR is the MC1PS ratio for False
      Positives and TP\_MR is the MR for True Positives. }
    \label{tab:simul2}
  \end{center}
\end{table}

Confirming the conclusions from Table~\ref{tab:simul2}, we  can
observe in Figures~\ref{fig:C1P} and~\ref{fig:MCS}, the following
facts.
\begin{itemize}
\item The conflicting ratio (CR) discriminates effectively between
  false positives and true positives.  For example conserving all
  syntenies that have a CR at least $0.14$ results in discarding
  $80\%$ of FP while still keeping $83\%$ of TP.
\item The MC1PS ratio (C1PR) does not discriminate as effectively
  between false positives and true positives.
\end{itemize}
It is also interesting to note that only $31\%$ of true positives
ancestral syntenies do not belong to any MCS. These suggests that, at
least with these simulated data, a significant number of ancestral
syntenies do not need to be considered when trying to detect false
positives (as we expect rows with CI equal to $0$ to be true positives
in general), but also that a very large part of true positive show
some conflicting signal despite the low ratio of false positives. This
shows that the extreme approach of discarding all rows belonging to at
least one MCS, suggested in~\cite{bergeron-reconstructing}, can result
in discarding a very large number of true positives.

\begin{figure}
  \begin{center}\scalebox{0.45}{\includegraphics{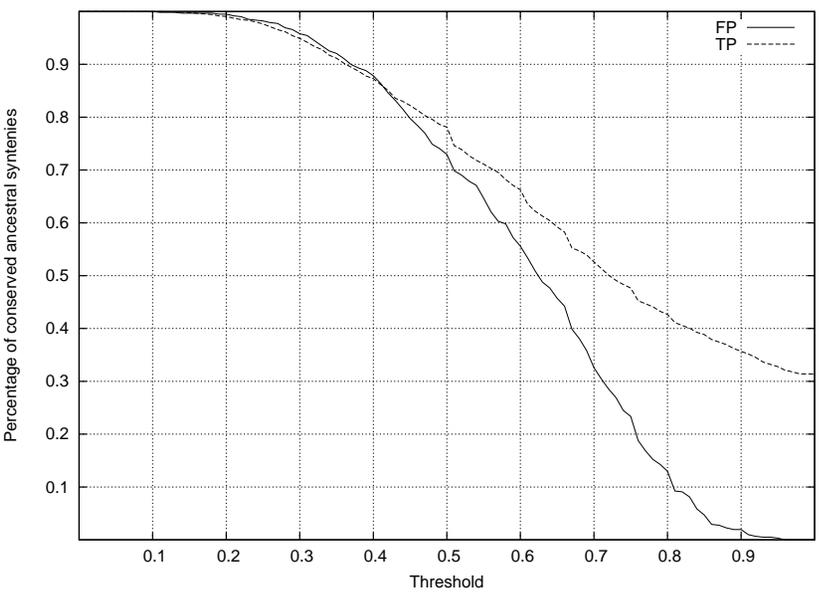}}\end{center}
  \caption{Percentage of conserved ancestral syntenies (y-axis) given
    a minimal C1PR (Threshold, x-axis), for both false positives (FP)
    and true positives (TP) for simulated datasets.}
  \label{fig:C1P}
\end{figure}

\begin{figure}
  \begin{center}\scalebox{0.45}{\includegraphics{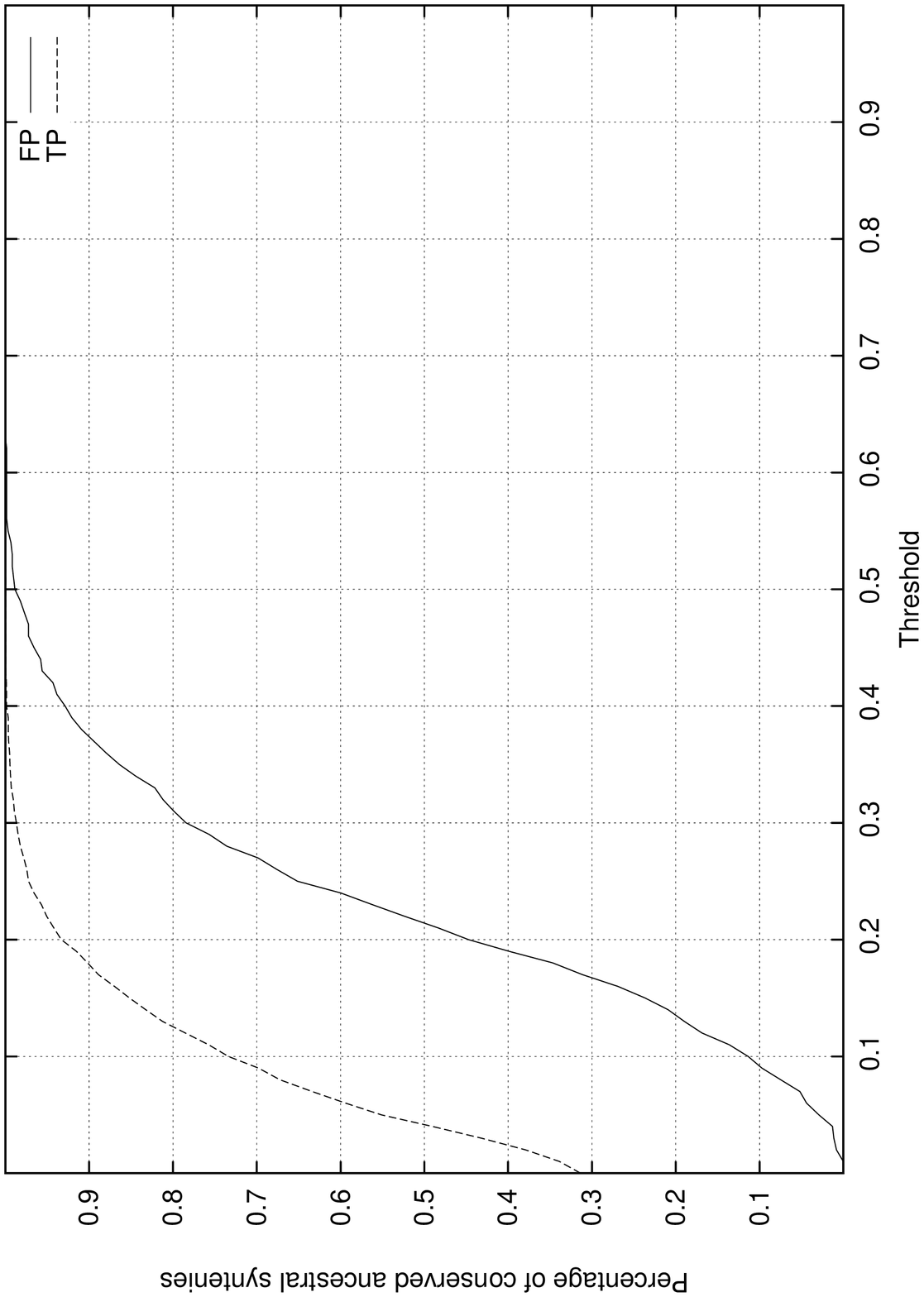}}\end{center}
  \caption{Percentage of conserved ancestral syntenies (y-axis) given
    a maximal C1PR (Threshold, x-axis), for both false positives (FP)
    and true positives (TP) for simulated datasets.}
  \label{fig:MCS}
\end{figure}

\subsection{Application on yeasts real data}
\label{sec:yeasts}

Next, we considered a real dataset, used to reconstruct an ancestral
{\em Saccahromycetaceae} genome from unduplicated yeasts genomes ({\em
  S. kluyveri}, {\em K. thermotolerans}, {\em K. lactis}, {\em
  A. gossypii} and {\em Z. rouxii}), described
in~\cite{chauve-yeasts}. These genomes are represented with $1420$
markers\footnote{The $1420$ markers represent $710$ synteny blocks,
  each block being split into two markers corresponding to its two
  extremities and and represented by an ancestral synteny of size $2$
  containing these two markers.}, and a total of $3106$ ancestral
syntenies were computed, giving a binary matrix $M$ with $3106$ rows
and $1420$ columns. From this large matrix, five submatrices
$\{M_1,\ldots,M_5\}$ were extracted that contained all MCS (each
matrix corresponds to an R-node of the PQR-tree associated to $M$,
see~\cite{chauve-methodological,mcconnell-certifying} for
details). For these five matrices, the number $m_i$ or rows and $n_i$
of columns are the following: $m_1=12$, $n_1=8$, $m_2=200$, $n_2=104$,
$m_3=262$, $n_3=126$, $m_4=801$, $n_4=348$, $m_5=1393$ and $n_5=652$.
For matrices $M_1$, $M_2$ and $M_3$, the joint generation computation
was completed within a few hours. For the larger matrices $M_4$ and
$M_5$, the computations were each interrupted after three days.  We
first report in Table~\ref{tab:yeasts1} the number of MCS and MC1PS 
detected\footnote{For each matrix, we filtered the set of MC1PS
  to discard any set of rows that does not contain all ancestral
  syntenies corresponding to the synteny blocks represented by the
  columns of this matrix.}. We also show the number of ancestral
syntenies that belong to all MC1PS, that we call {\em reliable
  ancestral syntenies}.

\begin{table}\small
  \begin{center}
    \begin{tabular}{|c|c|c|c|}
      \hline
      Matrix & Number of MCS & Number of MC1PS & Number of reliable ancestral syntenies
      \\ \hline 
      $M_1$  & $6$      & $4$      & $6$
      \\ \hline
      $M_2$  & $402$    & $13$     & $166$
      \\ \hline
      $M_3$  & $2214$   & $90$     & $199$
      \\ \hline
      $M_4$  & $1976^*$ & $483^*$  & $617^*$
      \\ \hline
      $M_5$  & $1277^*$ & $883^*$  & $1291^*$
      \\ \hline
    \end{tabular}
    \caption{Statistics on MCS and MC1PS on the yeasts
      dataset. Numbers marked by the symbol $^*$ correspond to
      partial results for interrupted computations.}
    \label{tab:yeasts1}
  \end{center}
\end{table}

We can observe that, unlike with simulated data, the number of MCS is
larger than the number of MC1PS, a fact that is not related to the
filtering of MC1PS. This could be explained by the fact that ancestral
syntenies can be much larger here than in the simulated data (where
they were of size at most $5$) which can imply that a single false
positive can create MCS with many different sets of true positive
rows. The second important observation is that approximately $85\%$
of all ancestral syntenies belong to all MC1PS (or at least, for $M_4$
and $M_5$, all generated MC1PS) and can then be considered as
reliable. This fact addresses a question that is raised
in~\cite{chauve-yeasts} about the reliability of ancestors computed
from ancestral syntenies, as it shows that most ancestral syntenies
are reliable, at least from a combinatorial optimization point of
view. To refine this observation, we show in
Figure~\ref{fig:yeasts} that most ancestral syntenies have a
high C1PR and a low CR.

\begin{figure}
  \begin{center}\scalebox{0.45}{\includegraphics{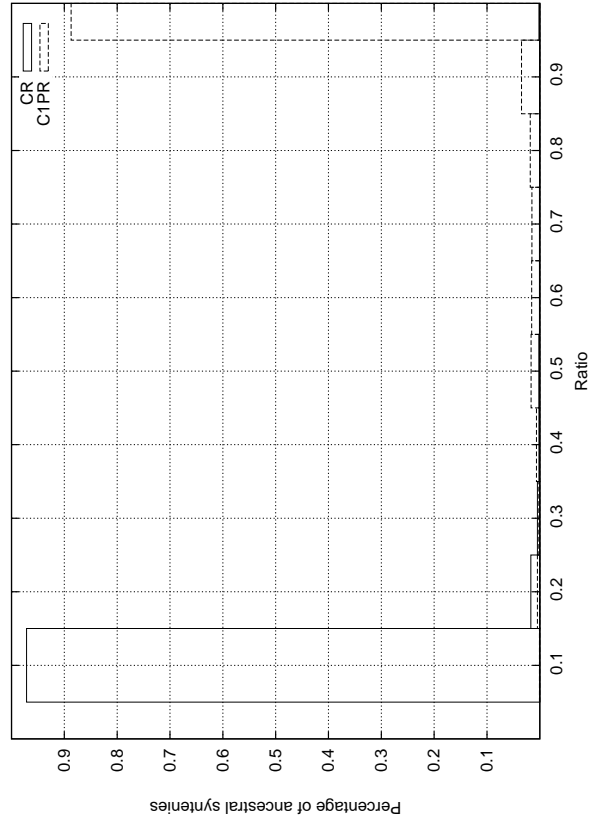}}\end{center}
  \caption{Percentage of conserved ancestral syntenies (y-axis) with a
    given CR (resp. C1PR). Each bar represents the percentage of
    ancestral syntenies whose CR (resp. C1PR) is in an interval of
    length $0.1$.}
  \label{fig:yeasts}
\end{figure}

Finally, we selected the subset of ancestral syntenies with a CR at
most $0.5$ and a C1PR at least $0.5$. This subset of ancestral
syntenies does not have the C1P, but only three ancestral syntenies
need to be discarded to obtain a C1P matrix, that defines a set of
$13$ CARs. In comparison, the set of all ancestral syntenies from
$M_1,\ldots, M_5$ required $12$ ancestral syntenies to be discarded in
order to have the C1P, leading to $9$ CARs. This shows that most CARs
obtained in~\cite{chauve-yeasts} are supported if only ancestral
syntenies with low CR and high C1PR are conserved.

\section{Conclusion and perspectives}
\label{sec:conc}

This paper describes preliminary theoretical and experimental results
on Minimal Conflicting Sets and Maximal C1P Sets. In particular, we
suggested that Tucker patterns are fundamental to understanding the
combinatorics of MCS, and that the generation of all MCS is a hard
problem, related to monotone boolean functions. From an experimental
point of view it appears, at least on datasets of adjacencies, that
MCS offer a better way to detect false positive ancestral syntenies
than MC1PS. From a methodological point of view, it suggests that the
joint generation framework provides a very general and flexible tool
for handling the notion of minimal conflicting data in computational
biology, as shown for example in~\cite{haus-knock}. This leaves
several open problems to attack.

\paragraph{Detecting non-conflicting rows.}
The complexity of detecting rows of a matrix that do not belong to any
MCS when rows can have an arbitrary number of entries $1$ is still
open. Solving this problem probably requires a better understanding of
the combinatorial structure of MCS and Tucker patterns.  Tucker
patterns have also be considered in~\cite[Chapter 3]{dom-recognition},
where polynomial time algorithms are given to compute a Tucker pattern
of a given type for a matrix that does not have the C1P. Even if these
algorithms can not obviously be modified to decide if a given row
belongs to a given Tucker pattern, they provide useful insight on
Tucker patterns.

It follows from the dual structure of monotone boolean functions that
the question of whether a row belongs to any MCS is equivalent to the
question of whether it belongs to any MC1PS. Indeed, for an arbitrary
oracle-given function, testing if a variable appears in any MTC is as
difficult as deciding if a list of MTC is complete.  Consider an
oracle-given $f$ and a list of its MTC which define a (possibly
different) function $f'$.  We can build a new oracle function $g$ with
an additional variable $x_0$, such that $g(x_0,x)=1$ if and only if
$x_0=0$ and $f'(x)=1$ or $x_0=1$ and $f(x)=1$.

\paragraph{Generating all MCS and MC1PS.} Right now, this can be
approached using the joint generation method, but the number of MCS
and MC1PS makes this approach time consuming for large matrices. A
natural way to deal with such problem would be to generate at random
and uniformly MCS and MC1PS. For MCS, this problem is at least as hard
as generating random cycles of a graph, which is known to be a hard
problem~\cite{jerrum-random}. We are not aware of any work on the
random generation of MC1PS.

An alternative to random generation would be to abort the joint
generation after it generates a large number of MCS and MC1PS, but the
quality of the approximation of the MCS ratio and MC1PS ratio so
obtained would not be guaranteed.  Another approach for the generation
of all MCS is based on the remark that, for adjacencies, it can be
reduced to generating all claws and cycles of the graph
$G_M$. Generating all cycles of a graph can be done in time that is
polynomial in the number of cycles, using
backtracking~\cite{read-bounds}. It is then tempting to use this
approach in conjunction with dynamic partition
refinement~\cite{habib-lex} for example or the graph-theoretical
properties of Tucker patterns described in~\cite{dom-recognition}.

\paragraph{Combinatorial characterization of false positive ancestral
  syntenies.}  It is interesting to remark that, with matrices of
degree $2$, most false positives can be identified in a simple
way. True positive rows define a set of paths in the graph $G_M$,
representing ancestral genome segments, while false positive rows
$\{i,j\}$, unless $i$ or $j$ is an extremity of such a path (in which
case it does not exhibit any combinatorial sign of being a false
positive), both the vertices $i$ and $j$ belong to a claw in the graph
$G_M$. And it is easy to detect all edges in this graph with both ends
belonging to a claw. In order to extend this approach to more general
datasets, where $\Delta(M) > 2$, it would be helpful to understand
better the impact of adding a false positive row in $M$. The most
promising approach would be to start from the {\em partition
  refinement}~\cite{habib-lex} obtained from all true positive rows
and form a better understanding of the combinatorial structure of
connected components of the overlap graph that do not have the C1P.

\paragraph{Computation speed.}  
On large datasets, especially with matrices with an arbitrary number
of entries $1$ per row, some connected components of the overlap graph
can be very large (see the data in~\cite{chauve-yeasts} for
example). In order to speed up the computations, algorithmic design
and engineering developments are required, both in the joint
generation algorithm and in the problem of testing the C1P for
matrices after rows are added or removed.


\section{Acknowledgments}
Cedric Chauve and Tamon Stephen were partially supported by NSERC
Discovery Grants.  Utz-Uwe Haus was supported by the Magdeburg Center
for Systems Biology, funded by a FORSYS grant of the German Ministry
of Education and Research.  Vivija P. You was partially supported by
The Department of Mathematics of Simon Fraser University. The authors
would also like to acknowledge the support of the IRMACS Centre at
SFU.  A preliminary version of this paper appeared as
\cite{chauve-v0}.


\bibliographystyle{jcb}

~\vfill\pagebreak
\section*{Appendix: the five Tucker patterns}

$$
M_{I} :
\begin{array}{|r|cccccccc|}
 \hline
 & c_1 & c_2 & c_3 & c_4 & \dots & c_{q} & c_{q+1} & c_{q+2}\\
 \hline
 r_1 & \mathbf{1} & \mathbf{1} & {0} & {0} & \dots & {0} & {0} & {0} \\
 r_2 & {0} & \mathbf{1} & \mathbf{1} & {0} & \dots & {0} & {0} & {0} \\  
 r_3 & {0} & {0} & \mathbf{1} & \mathbf{1} & \dots & {0} & {0} & {0} \\
\vdots& \vdots     &  \vdots       &     \vdots    &      \vdots   &\ddots&     \vdots    &    \vdots    &  \vdots\\
 r_{q} & {0} & {0} & {0} & {0} & \dots & \mathbf{1} & \mathbf{1} & {0} \\  
r_{q+1}& {0} & {0} & {0} & {0} & \dots & {0} & \mathbf{1} & \mathbf{1} \\
r_{q+2}& \mathbf{1} & {0} & {0} & {0} & \dots & {0} & {0} & \mathbf{1} \\ 
 \hline
\end{array}
$$

$$
M_{II} :
\begin{array}{|r|ccccccccc|}
 \hline
 & c_1 & c_2 & c_3 & c_4 & \dots & c_{q} & c_{q+1} & c_{q+2} & c_{q+3}\\
 \hline
 r_1 & \mathbf{1} & \mathbf{1} & {0} & {0} & \dots & {0} & {0} & {0} & {0}\\
 r_2 & {0} & \mathbf{1} & \mathbf{1} & {0} & \dots & {0} & {0} & {0} & {0}\\  
 r_3 & {0} & {0} & \mathbf{1} & \mathbf{1} & \dots & {0} & {0} & {0} & {0}\\
\vdots& \vdots     &  \vdots       &     \vdots    &      \vdots   &\ddots&     \vdots    &    \vdots    &  \vdots        &  \vdots \\
 r_{q} & {0} & {0} & {0} & {0} & \dots & \mathbf{1} & \mathbf{1} & {0} & {0}\\  
r_{q+1}& {0} & {0} & {0} & {0} & \dots & {0} & \mathbf{1} & \mathbf{1}& {0} \\
r_{q+2}& \mathbf{1} & \mathbf{1} & \mathbf{1} & \mathbf{1} & \dots & \mathbf{1} & \mathbf{1} & {0}& \mathbf{1} \\ 
r_{q+3}& {0} & \mathbf{1} & \mathbf{1} & \mathbf{1} & \dots & \mathbf{1} & \mathbf{1} & \mathbf{1}& \mathbf{1} \\ \hline
\end{array}
$$

$$
M_{III} :
\begin{array}{|r|ccccccccc|}
 \hline
 & c_1 & c_2 & c_3 & c_4 & \dots & c_{q} & c_{q+1} & c_{q+2} & c_{q+3}\\
 \hline
 r_1 & \mathbf{1} & \mathbf{1} & {0} & {0} & \dots & {0} & {0} & {0} & {0}\\
 r_2 & {0} & \mathbf{1} & \mathbf{1} & {0} & \dots & {0} & {0} & {0} & {0}\\  
 r_3 & {0} & {0} & \mathbf{1} & \mathbf{1} & \dots & {0} & {0} & {0} & {0}\\
\vdots& \vdots     &  \vdots       &     \vdots    &      \vdots   &\ddots&     \vdots    &    \vdots    &  \vdots        &  \vdots \\
 r_{q} & {0} & {0} & {0} & {0} & \dots & \mathbf{1} & \mathbf{1} & {0} & {0}\\  
r_{q+1}& {0} & {0} & {0} & {0} & \dots & {0} & \mathbf{1} & \mathbf{1}& {0} \\
r_{q+2}& {0} & \mathbf{1} & \mathbf{1} & \mathbf{1} & \dots & \mathbf{1} & \mathbf{1} & {0}& \mathbf{1} \\ 
\hline
\end{array}
$$

$$
M_{IV} :
\begin{array}{|r|cccccc|}
 \hline
 & c_1 & c_2 & c_3 & c_4 & c_5 & c_6\\
 \hline
 r_1& \mathbf{1} & \mathbf{1} & {0} &{0} &{0} &  {0} \\
 r_2 & {0} & {0} & \mathbf{1}& \mathbf{1} &{0} &{0} \\
 r_3 & {0} & {0} &{0} & {0} &\mathbf{1} & \mathbf{1} \\
 r_4 & {0} & \mathbf{1} &{0} & \mathbf{1} &{0} & \mathbf{1} \\ 
\hline
\end{array}
$$

$$
M_{V} :
\begin{array}{|r|ccccc|}
 \hline
 & c_1 & c_2 & c_3 & c_4 & c_5 \\
 \hline
 r_1& \mathbf{1} & \mathbf{1} & {0} &{0} &  {0} \\
 r_2 & \mathbf{1} & \mathbf{1} & \mathbf{1}& \mathbf{1} &{0}  \\
 r_3 & {0} &{0} & \mathbf{1} &\mathbf{1} & {0} \\
 r_4 & \mathbf{1} & {0} &{0} & \mathbf{1} &\mathbf{1}  \\ 
\hline
\end{array}
$$

\end{document}